\documentclass[conference,a4paper]{IEEEtran}
\usepackage{hyperref}
\usepackage{enumerate}

\usepackage[utf8]{inputenc}

\makeatletter\@addtoreset{chapter}{part}\makeatother%
\usepackage{cite}

\usepackage{graphicx}
\graphicspath{ {images/} }
\usepackage{acronym}

\usepackage{datetime}
\usepackage{algorithm}

\usepackage[normalem]{ulem}
\usepackage{geometry}
\usepackage[normalem]{ulem}
\usepackage{makeidx}
\makeindex
\usepackage{setspace}
\newboolean{bol_FULL_ARTICLE_VERSION}
\setboolean{bol_FULL_ARTICLE_VERSION}{true} 
\newboolean{bol_SHOWCOLORS}
\setboolean{bol_SHOWCOLORS}{true}

\newcommand{\OptionalColorBlue}{\ifthenelse{\boolean{bol_SHOWCOLORS}}{\color{ blue}}{\color{black}}}

\newcommand{\OptionalColorViolet}{\ifthenelse{\boolean{bol_SHOWCOLORS}}{\color{violet}\color{black}}{\color{black}}}

\usepackage{bbm}

\usepackage{amssymb,amsfonts,amsmath,amstext,mathrsfs}
\DeclareMathOperator*{\argmax}{argmax} 

\usepackage{stmaryrd}
\usepackage{latexsym}
\usepackage{epsfig,psfrag,color,graphics,epsf}
\usepackage{xcolor}
\usepackage{enumerate,cite}
\usepackage{amsmath}
\usepackage{amssymb}
\usepackage{amsfonts}
\usepackage{amsthm}
\usepackage{graphicx}
\usepackage{float}
\usepackage{graphics}
\usepackage{graphicx}
\usepackage{euscript}
\usepackage{psfrag}
\usepackage{soul}

\usepackage{standalone}
\usepackage{tikz,pgfplots}
\usepackage{mathtools}

\newcommand{\reals} {\mathbb{R}}

\newcommand{\beq} {\begin{equation}}
\newcommand{\eeq} {\end{equation}}
\newcommand{\beqa} {\begin{align}}
\newcommand{\eeqa} {\end{align}}

\newcommand{\spll}{\nonumber \\&}

\usepackage{verbatim}

\newtheorem{theorem}{Theorem}
\newtheorem{lemma}{Lemma}

\newtheorem{definition}{Definition}

\newcommand {\bx} {\mbox{\boldmath $x$}}
\newcommand {\by} {\mbox{\boldmath $y$}}
\newcommand {\bz} {\mbox{\boldmath $z$}}

\newcommand {\bX} {\mbox{\boldmath $X$}}

\newcommand {\bZ} {\mbox{\boldmath $Z$}}

\newcommand{\calC}{{\cal C}}
\newcommand{\calD}{{\cal D}}

\newcommand{\calK}{{\cal K}}

\newcommand{\calN}{{\cal N}}

\newcommand{\calU}{{\cal U}}

\newcommand{\calX}{{\cal X}}
\newcommand{\calY}{{\cal Y}}

\newcommand{\EE}{{\mathbb E}}

\newcommand\mtiny[1]{\mbox{\tiny\ensuremath{#1}}} 
\allowdisplaybreaks

\usepackage{cjhebrew}
\makeatletter

\usepackage[utf8]{inputenc} 
\usepackage[T1]{fontenc}
\usepackage{url}              
\usepackage{cite}             

\usepackage{mleftright}       
\mleftright                   

\usepackage{graphicx}         
\usepackage{booktabs}         

\usepackage{algorithmicx}    

\begin{document}
\title{Pre-Decoder Processing Functions for a DMC with Mismatched Decoding} 


\author{%
    \IEEEauthorblockN{Jonathan Solel and Anelia Somekh-Baruch}
    \IEEEauthorblockA{Faculty of Engineering,
                    Bar-Ilan University\\
                    Ramat-Gan, Israel\\
                    Email: yonis.mailbox@gmail.com, somekha@biu.ac.il}
}

\maketitle
\begin{abstract}
 This paper analyzes the effect of adding a pre-decoder processing
function to a receiver that contains a fixed mismatched decoder at the output of a discrete memoryless channel. 
We study properties of the symbolwise pre-processing function and show that it is a simple yet very powerful tool which enables to obtain reliable transmission at a positive rate for almost every metric. 
We present lower and upper bounds on the capacity of a channel with mismatched decoding and symbolwise(scalar-to-scalar) pre-processing, and show that the optimal pre-processing function for random coding is deterministic. We also characterize achievable error exponents. Finally, we prove that a separation principle holds for vectorwise(vector-to-vector) pre-processing functions and further, that deterministic functions maximize the reliably transmitted rate in this case. 
\let\thefootnote\relax\footnotetext{
This work was supported by the Israel Science Foundation (ISF) under grant \#1579/23.}

\end{abstract}

\section{Introduction}
In the classical channel coding problem it is assumed that in order to maximize the transmitted information rate for a certain channel one can choose an encoder-decoder pair that minimizes the average error probability.
 However, in many cases one may not be able to use the optimal decoder due to practical reasons such as computational or design limitations, and generally not having perfect knowledge of the channel parameters. In those cases one will have to use a sub-optimal decoder instead of the optimal maximum likelihood (ML) decoder w.r.t the channel. Using certain sub-optimal decoders might result in lower achievable rates, the supremum over these rates is referred to as the mismatch capacity of the channel.\\
Extensive research in Information Theory has explored channels featuring a decoder mismatch, providing derivations of lower and upper bounds on the mismatch capacity and error exponents
(see e.g. \cite{Lapidoth93,hui1983fundamental,CsiszarKorner,scarlett2015counter,Lapidoth96,CsiszarNarayan95,SomekhBaruch_mismatchachievableIT2014,ScarlettMartinezFabregas2016,kangarshahi2020singleletter,SomekhBaruch2022singleletter,Anelia_robust21,Sphere_erexp_23,SomekhBaruch2023UpperBoundGenie,somekh2023upper,somekhbaruch2022upper,BigBook}). 
In this paper we address a very practical setup. 
We consider a system which attempts to improve the performance of a given decoder by adding a pre-processing unit prior to it. Simple examples of pre-processing include  scaling and phase shifting. 
 \color{blue}\color{black}\color{black}
This situation may arise as a result of either design constraints or limitations on user access for modifying the receiver. In such cases, exploring alternative solutions other than replacing the receiver becomes crucial.
\color{black}
 A special case of this problem was analyzed in \cite{wang2022generalized} for the Gaussian channel with a mismatched generalized nearest neighbour decoding rule (GNNDR).
A more recent paper \cite{dikshtein2023mismatched}, considered the problem of reliable communication in a point-to-point oblivious-relay communication system with a mismatch at either the relay or at the decoder. The setup considered in \cite{dikshtein2023mismatched} includes three nodes, an encoder, a relay encoder and a mismatched decoder, 
 the relay encoder performs collective processing for all symbols associated with a particular message. Achievability bounds are established in\cite{dikshtein2023mismatched} and a computationally efficient algorithm to compute these bounds is presented.
 
While the setting in \cite{dikshtein2023mismatched} for the relay channel can be viewed as an extension of the vectorwise scheme presented in this paper (see Section \ref{sc: vpf}), the setting provided in Section \ref{sc: Ach} of this paper yields the same achievability bound as in \cite{dikshtein2023mismatched} \color{blue}\color{black}\color{black}for a high relay rate, using \color{black} a simpler scheme that is more latency and memory efficient.

\section{General Problem Formulation}\label{sc: Problem Formulation}
The basic setup discussed in this paper considers communication over a discrete memoryless channel (DMC) defined by a conditional probability distribution $W$ from $\calX$ to $\calY$, which are given finite sets.
The input-output probabilistic relation for an $n$-length codeword is given by:
\begin{flalign}\label{py_gv_x}
 P \left( \by | \bx \right) = W^n \left( \by|\bx\right) = \prod_{i=1} ^{n} ( W\left(y_i|x_i\right) ),
 \end{flalign}
 where  $\bx = \left(x_1 \dots x_n \right) \in \calX^n$ and $\by = \left(y_1 \dots y_n \right) \in \calY^n$ .
The setup consists of an encoder that maps each message $m \in \left\{  1, \dots,M_n \right\}$ to a channel input sequence $\bx_m \in \calX^n$,
and a decoder $\phi_{n,q}$ that maximizes the accumulated value of 
$q(\bx,\by) \triangleq \frac{1}{n}  \sum_{i=1} ^{n} q(x_i,y_i),$
 where $q:\calX \times \calY \longrightarrow \mathbb{R}$ is a mapping referred to as a metric.\\
 For each channel output $\by \in \calY ^n$, the decoder outputs an estimate $\widehat{m}_q \in \left\{  1, \dots,M_n \right\}$ to the transmitted message $m$:
\beq\label{eq:argmax}
\widehat{m}_q(\by) = \argmax_{k\in \{1,\dotsc,M_n\}} q^{}(\bx_k,\by),
\eeq
An error occurs when the decoded message differs from the transmitted one or when there are several maximizers in (\ref{eq:argmax}).
It is easily verified due to (\ref{py_gv_x}) and the monotonicity of the $\log(\cdot)$ function, that the ML decoding metric is given by
$
q_{\mtiny{ML}}(x,y)= \log W(y|x).
$
 \\In this paper we assume that the receiver can add a component prior to the existing $q$-decoder, and therefore is in fact composed of two sequential components:
\\(a) A (possibly stochastic) memoryless configurable symbolwise function $f$ from $\calY$ to $\calY$.
\\(b) The decoder of the form (\ref{eq:argmax}) with a given metric $q$.
\\Let us denote by $Z_i$ the output of the pre-processing function $f$ when fed by the channel output $Y_i$ at time instant $i$.
This scheme is depicted in Figure\ \ref{fig:channel_scheme}:
\begin{figure}[hbt!]
  \centering
  \includegraphics[width=0.48\textwidth]{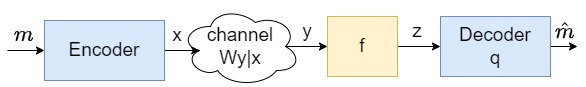}
    \caption[The channel scheme including pre-processing ]{Channel with a mismatched decoder and pre-processing function}
    \label{fig:channel_scheme}
\end{figure}
\\Note that the resulting average probability of error is given by:
\begin{flalign} 
P_e(W,&\calC_n,q,f)= 
\sum_{m=1}^{M_n}\frac{1}{M_n} P(\hat{m}_q(\bz) \neq m|\bX=\bx_m), \nonumber
\end{flalign}
where
$\bz=(z_1,\ldots,z_n)$, and $z_i=f(y_i)$. 
\begin{definition}
A rate $R>0$ is said to be achievable with pre-processing for channel $W$ with decoding metric $q$ if there exists a sequence of codes $\{\calC_n\}_{n\in \mathbb{N} }$, and a possibly stochastic transformation $f$ from $\calY$ to $\calY$
such that $|\calC_n|>e^{nR}$ and the average probability of error vanishes; i.e., 
\begin{flalign}
\lim_{n\rightarrow\infty} \min_{f} P_e(W,\calC_n,q,f)=0.
\end{flalign}
\end{definition}
\begin{definition}
The mismatch capacity of channel $W$ with 
decoding metric $q$ and pre-processing, denoted $C_q^{pre}(W)$, is the supremum of all achievable rates \color{blue}\color{black}\color{black}with pre-processing.\color{black} 
\end{definition}{}
An achievable error exponent and the corresponding reliability function for this case are defined as follows:
 \begin{definition}
For a DMC with a mismatched decoder defined by $(W, q)$ and a symbolwise pre-processing function $f$, we say that $E$ is an achievable error exponent at rate $R$ if, for any $\delta > 0$, there exists a sequence of codebooks $\calC_n$ satisfying
$
| \calC_n | \geq e^{n(R-\delta)}
$
such that
\begin{flalign}
    \limsup_{n\to\infty}\left(- \frac{1}{n} \log P_e(W, \calC_n,q, f) \right) \geq E.
\end{flalign}
\end{definition}
\begin{definition}\label{def:rel_fun}
The reliability function $E_q^{pre}(R,W)$ at rate $R$ for a DMC $W$ with decoding metric $q$ and pre-processing is the supremum over  all achievable error exponents at rate $R$ and all pre-processing functions $f$.
\end{definition}

\section{The relationship between the quantities $C_q^{pre}(W)$, $C_q(W)$, and $C(W)$}
We start by solving a few examples that illustrate different scenarios regarding the relationships between the three quantities $C_q^{pre}(W)$, $C_q(W)$, and the Shannon capacity $C(W)$. Further, we address the question of determining when the use of a pre-processing function can be beneficial. 
We demonstrate the extreme (and rather common) case in which one has $C_q(W)=0$ while $C_q^{pre}(W)>0$; i.e., the pre-processing makes it possible to transmit information over a channel with mismatched decoding in which it was originally impossible to do so. 
To this end, we introduce the notion of a useless metric, and analyze conditions for a metric to be useless. 

\subsection{Examples}\label{sc: Examples}
\subsubsection{A Binary Channel}\label{sbsc: example_A}
		This example demonstrates that a pre-processing function can potentially achieve $C_q^{pre}(W) > C_q(W)  $.
 Assume a binary channel;  i.e., $|\mathcal{X}| = |\mathcal{Y}| = 2$, with channel transition probabilities $W(i|j)$, $i,j\in\{0,1\}$. 
It was shown in \cite[Example 6]{CsiszarNarayan95} that 
for a binary input binary output channel $W$, it holds that if $W(0|1)+W(1|0) < 1$ then \begin{flalign}\label{eq:bin_chan1}
  C&_q(W) = \spll
 \begin{cases} 0 &\mbox{if } q(0,1)+q(1,0) \geq q(0,0) +q(1,1) \\
C(W) &\mbox{if } q(0,1)+q(1,0) < q(0,0) +q(1,1) \end{cases} 
\end{flalign}
and if $W(0|1)+W(1|0) > 1$ then  
\begin{flalign}\label{eq:bin_chan2} &C_q(W)= \spll
\begin{cases} C(W) &\mbox{if } q(0,1)+q(1,0) > q(0,0) +q(1,1) \\
0 & \mbox{if } q(0,1)+q(1,0) \leq q(0,0) +q(1,1) \end{cases}.\end{flalign}

Thus, if we define a pre-processing function 
as: 
if $W(0|1)+W(1|0) < 1$ then  \begin{flalign} f(y) = \begin{cases} 1-y &\mbox{if } q(0,1)+q(1,0) \geq q(0,0) +q(1,1) \\
y & \mbox{if } q(0,1)+q(1,0) < q(0,0) +q(1,1) \end{cases} \nonumber\end{flalign}
 and if $W(0|1)+W(1|0) > 1$ then  \begin{flalign} f(y) = \begin{cases} y &\mbox{if } q(0,1)+q(1,0) > q(0,0) +q(1,1) \\
1-y & \mbox{if } q(0,1)+q(1,0) \leq q(0,0) +q(1,1) \end{cases} \nonumber\end{flalign}
we clearly obtain
\begin{flalign} C_q^{pre}&(W) =\spll \begin{cases} C(W) &\mbox{if } q(0,1)+q(1,0) \neq q(0,0) +q(1,1) \\
0 & \mbox{if } q(0,1)+q(1,0) = q(0,0) +q(1,1) \end{cases} .\nonumber \end{flalign} 
It is interesting to observe that this case results in almost every metric achieving the channel Shannon capacity. 
\subsubsection{A DMC with ``the worst" metric}\label{sbsc: example_ADD}
Generalizing on the previous example, consider the case of a channel $W$ with metric $q(x,y)=-q_{\mtiny{ML}}(x,y)=-\log W(y|x)$ where $\calX=\calY=\{-k,-k+1,...,0,...,k\}$ for some integer $k>0$. Clearly this mismatched decoder errs with very high probability and $C_q(W)=0$. One can easily fix this using the pre-processing function $f(y)=-y$, which implies that $C_q^{pre}(W)=C(W)$. 
\subsubsection{Deterministic Quadratic Channel 1}\label{sbsc: example_B}
The following example demonstrates
a case for which $$C_q(W) = C_q^{pre}(W) < C(W).$$
Assume that $|\calX| = |\calY| =  4$, with the noiseless channel
$ W(y|x) = \mathbbm{1}\big\{ (x,y) \in \left\{(0,0),(1,1),(2,2),(3,3)\right\} \big\}$ where clearly $C(W)=\log 4$. 
Consider the decoding metric:
$q(x,y) = \mathbbm{1}\big\{ (x,y) \in \left\{(0,0),(1,1),(2,2),(3,2)  \right\} \big\}. $
Evidently, 
transmission of symbols $x=2$ and $x=3$ always results in the same metric value. Thus, for any input signals pair ${\bx}_1,{\bx}_2$ s.t ${\bx}_2$ is same as ${\bx}_1$ with the exception of some of the symbols switched from $2$ to $3$ and vice versa, one has $q({\bx}_1,y) = q({\bx}_2,y)$.
Hence, the channel with the decoding metric $q$ is similar to a noiseless ternary channel and $C_q(W)=\log 3$. It is easy to realize that pre-processing cannot fix the decoder's metrics errors, and also $C_q^{pre}(W)=\log 3$.
 
\subsubsection{A Deterministic Quadratic Channel 2}\label{sbsc: example_C}
	The following example demonstrates
 a case for which $$C_q(W) < C_q^{pre}(W) < C(W).$$
Assume that $|\calX| = |\calY| =  4$, with a channel conditional probability $W(y|x) = \mathbbm{1}\{y=(x+2)\mod 4\}$. 
Consider the decoding metric: 
$q(x,y) = \mathbbm{1}\big\{ (x,y) \in \left\{(0,0),(1,1),(2,2),(3,2)  \right\} \big\}. $
Using this decoder no received signal is correctly decoded, thus $C_q(W)= 0$.
However, using pre-processing function $P(z|y)=\mathbbm{1}\{z=(y+2)\mod4\}$,
we get $P(z|x)=\mathbbm{1}\{z=x\}$. 
And thus, as seen in the previous example, we obtain $C_q^{pre}(W)=  \log3$,
demonstrating that, for this example $C_q(W)< C_q^{pre}(W)< C(W).$	

\ifthenelse{\boolean{bol_FULL_ARTICLE_VERSION}}{\color{violet}\color{black}
\subsubsection{A Discrete Input Additive Gaussian Noise Channel}\label{sccc: example_D}

Although the primary focus of this paper revolves around DMCs, this example shows the potential benefit of using pre-processing functions in continuous channels, by demonstrating a setup in which $0=C_q(W)<C_q ^{pre} (W)$.
Let $\calX = \left\{ (0,1), (0,-1) \right\}$, and consider a two dimensional signal $X\in\calX$.
Assume an additive Gaussian noise such that $Y= X +N$, where $N=(N_1 ,0)$ with $N_1 \sim \calN(0,1)$. Let $q(x,y) = -| y_1 -x_1 |^2$, where $x_1$ and $y_1$ denote the values of the first entries of $x$ and $y$ respectively.

Note that this decoder cannot distinguish between the symbols $(0,1)$ and $(0,-1)$, leading to $C_q (W) = 0$. However, by employing a pre-processing function that interchanges the coordinates order such that $z = (y_2, y_1)$ along with the provided decoder, we effectively utilize the Maximum Likelihood (ML) decoder for this channel. Consequently, $C_q ^{pre}(W) = C(W) = 1$.

\color{black}}{}
 
\subsection{Useless metrics}\label{sc: um}

This section provides necessary and sufficient 
conditions under which $C_q^{pre}(W)=0$ for all channels $W$; that is, when even pre-processing cannot correct a poor choice of decoding metric $q$ regardless of the channel. To this end, we introduce the following definition:
\begin{definition}
    We say that a metric $q$ is {\it useless for channel} $W$ with pre-processing if $C_q^{pre}(W) = 0 $. 
    We say that a metric $q$ is useless, if for all $W$, $C_q^{pre}(W)=0$. 
\end{definition}
The following theorem specifies conditions for a metric to be useless.
 \begin{theorem}\label{th:useless}
There exists a channel $W$ from $\calX$ to $\calY$ s.t $C_q^{pre}(W) >0$ if and only if $\exists x_1,x_2 \in \calX,  y_1,y_2 \in \calY$ s.t.
\begin{flalign} \label{eq:useless_cond}
q( x_1,y_1)-q( x_2,y_1) \neq  q( x_1,y_2)-q( x_2,y_2).
\end{flalign}
 \end{theorem}
 To prove this we use \cite[Example 6]{CsiszarNarayan95} to show that the condition in (\ref{eq:useless_cond}) is a sufficient and necessary for a metric to not be useless. See \ifthenelse{\boolean{bol_FULL_ARTICLE_VERSION}}{Appendix B }{\color{blue}\color{black}\color{black}\cite[Appendix B]{ThisArticle_arxiv}\color{black} }\ for the proof.

\ifthenelse{\boolean{bol_FULL_ARTICLE_VERSION}}{\color{violet}\color{black}
Let $R_{LM}(W)$ denote the LM rate achieved by random constant composition coding for a given channel; that is,
\begin{flalign}\label{def:RLM_W_np}
R_{LM}(W)&\triangleq \max_{\substack{P_{X}} } \min_{\substack{P_{X'|Y}:P_{X'}=P_X  \\ \EE q(X',Y)\geq \EE q(X,Y)}}I(X';Y), 
\end{flalign}
The following condition for a metric to be useless for a given channel, is derived directly from \cite[lemma 1]{CsiszarNarayan95}.
\begin{lemma}\label{th: useless for W}
    For a given channel $W$ from $\calX$ to $\calY$, $C_q^{pre}(W) >0$ iff there exists $P_{Z|Y}$ from $\calY $ to $\calY$ s.t $R_{LM}(P_{Z|X}) >0$, where $P_{Z|X}(z|x)= \sum_{y \in \calY} P_{Z|Y}(z|y)W(y|x)$.
\end{lemma}
\begin{proof}
Let $\calX,\calY$, $q:\calX\times\calY\rightarrow\mathbb{R}$ and $W$ from $\calX$ to $\calY$ be given.
Assume that $C_q^{pre}(W) >0$, thus there exists a (possibly stochastic) pre-processing function with transition probabilities $P_{Z|Y}$, s.t. $C_q(P_{Z|X}) >0$ where $P_{Z|X}(z|x)= \sum_{y \in \calY} P_{Z|Y}(z|y)W(y|x)$, and
according to \cite[lemma 1]{CsiszarNarayan95} for every channel $W$ the mismatch capacity satisfies $C_q(W)>0 \iff R_{LM}( W)>0$, 
thus, $R_{LM}(P_{Z|X}) >0$.
Now, $\forall P_{Z|Y}$, $C_q^{pre}(W) \geq R_{LM}(P_{Z|X})$, and  therefore if $C_q^{pre}(W)=0$ then $\forall P_{Z|Y} : R_{LM}(P_{Z|X}) = 0$.
\end{proof}
\color{black}}{}

\section{Achievable Rates and Error Exponents}\label{sc: Ach}
 This section discusses the use of a symbolwise pre-processing function $f$ for a DMC. We begin by presenting an achievability bound for the rate and error exponent using constant-composition random coding. In Section \ref{ap: Ach_sp} an improvement of the achievability bound is presented using superposition coding.
\subsection{Main Achievability Theorem}\label{ap: Ach_mt}
In this section a lower bound on $C_q^{pre} (W)$ 
is derived. 
Denote
\begin{flalign}\label{def:RLM_PX}
R^{pre}_{LM,q}(P_X,W)&\triangleq \max_{\substack{P_{Z|Y}} } \min_{\substack{P_{X'|Z}:P_{X'}=P_X  \\ \EE q(X',Z)\geq \EE q(X,Z)}}I(X';Z), 
\end{flalign} 
where $\EE(\cdot)$ denotes expectation, the maximization is over $P_{Z|Y}$ from $\calY$ to $\calY$, and the minimization is over $P_{X'|Z}$ from $\calX$ to $\calY$, satisfying the specified constraints, where $(X,Y,Z)\sim P_X\times W\times P_{Z|Y}$, i.e., $X-Y-Z$ is a Markov chain. 
The \color{blue}\color{black}\color{black}rate \color{black}$R^{pre}_{LM,q}(P_X,W)$ can be viewed as a lower bound on the achievable rate using a codebook with fixed composition \color{blue}\color{black}\color{black}$P_X$\color{black}. Further denote
\begin{flalign}\label{def:RLM}
R^{pre}_{LM,q}(W)&\triangleq \max_{\substack{ P_X } } R^{pre}_{LM,q}(P_X,W),
\end{flalign} 
and let
\begin{flalign}\label{def:ELM}
E&_{LM,q} ^{pre}(R,W) \triangleq\spll \max_{P_{Z|Y},P_X}\bigg( \min_{\substack{V_{X'Z}:E_{V_{X'Z}} \left[ q(X',Z) \right]\geq E_{V_{XZ}} \left[ q(X,Z) \right]\\V_{X'} = V_{X}=P_X  }} \spll \left[ D(V_{XZ}\| P_{XZ})+   \left| I_V(X';Z)-R \right|_+ \right] \bigg),
\end{flalign} 
where, similar to (\ref{def:RLM}),  $(X,Y,Z)\sim P_X\times W\times P_{Z|Y}$.
Here we use the shorthand notation for the function $|t|_+=\max\{0,t\}$. 
Our first result is the following. 
\begin{theorem}\label{th: ach1}
For a DMC $W$ and metric $q$ it holds that
\begin{flalign}
    &C_q^{pre}(W)\geq R^{pre}_{LM,q}(W),\label{eq: LM bound for pre-processing}\\
  &E_q^{pre}(R,W)\geq E_{LM,q}^{pre}(R,W)\label{eq: reliabiity bound for pre-processing}
  \end{flalign}
 \end{theorem}
 Note that adding a pre-processing function effectively creates a DMC from $X$ to $Z$,
and hence, the bounds (\ref{eq: LM bound for pre-processing})-(\ref{eq: reliabiity bound for pre-processing}) can be obtained quite straightforwardly from previous results (e.g.\ \cite{CsiszarKorner,hui1983fundamental}, and \cite[Theorem 2.3]{BigBook}) by substituting the channel $W$ with $P_{Z|X}$ and maximizing the obtained expression with respect to $P_{Z|Y}$.

Examining (\ref{eq: LM bound for pre-processing}) and constraining the pre-processing function $P_{Z|Y}$ to satisfy $I(Y;Z)\leq B$ results in an achievable rate:
\begin{flalign}\label{eq: LM bound with I<B}
\max_{\substack{P_{Z|Y}, P_X} } \min_{\substack{P_{X'|Z}:P_{X'}=P_X  \\ \EE q(X',Z)\geq \EE q(X,Z) \\ I(Y;Z)\leq B}} I(X';Z).
\end{flalign}
This achievable rate coincides with the achievability bound presented in  \cite{dikshtein2023mismatched} \color{blue}\color{black}\color{black}. Note that symbolwise pre-processing is an applicable scheme 
for the setup of \cite{dikshtein2023mismatched} in the case of high relay rate $B$. Therefore, symbolwise pre-processing offers a memory and latency efficient solution compared to a relay block function in this case. \color{black}


\subsection{Superposition Coding}\label{ap: Ach_sp}
The above mentioned argument following Theorem \ref{th: ach1}, can also be used to show that superposition coding can improve on Theorem \ref{th: ach1} and provide a tighter achievable bound on $ C_q^{pre}(W)$.
For example, using \cite[Theorem 7]{ScarlettMartinezFabregas2016} we get that for a given pre-processing function $P_{Z|Y}$, for any finite set $ \calU $ and input distribution $Q_{UX}$, and for a set of rates $ R_0(P_{Z|Y}), \{ R_{1,u}(P_{Z|Y}) \}_{u=1} ^{|\calU|}$, which depend on \ifthenelse{\boolean{true}}{the pre-processing function }{ }$P_{Z|Y}$,  
the rate 
\begin{flalign}R(P_{Z|Y})= R_0(P_{Z|Y}) +\sum_{u\in \calU}{Q_U(u)R_{1,u}(P_{Z|Y})} 
\end{flalign} 
is achievable provided that:
\begin{flalign}
    R_{1,u}&(P_{Z|Y})\leq R^{pre}_{LM,q}(Q_{X|U=u},P_{Z|U=u})
\end{flalign}
where $P_{XZ|U=u}(x,z|u) = \sum_{y \in \calY} Q_{X|U=u}(x|u) \cdot W(y|x) \cdot P_{Z|Y}(z|y)$, $P_{Z|X}(z|x) = \sum_{y\in \calY} W(y|x) \cdot P_{Z|Y}(z|y)$, and 
\begin{flalign}
R_0(P_{Z|Y})& \leq \min_{\substack{P_{U'X'Z'} \in \tau_0 (Q_{UX} \times P_{Z|X})}}{I(U';Z')} + \spll | \max_{\calK \subset \calU}{\sum_{u \in \calK}{Q_U(u)I(X';Z|U=u) - R_{1,u}}  } |_{+},\nonumber    
\end{flalign}
where $\tau_0$ is defined as 
\begin{flalign}
\tau_0 &(P_{UXZ}) = \bigg\{ {P}_{U'X'Z'} : {P}_{U'X'} = P_{UX}, {P}_Z' = P_Z, \spll \EE_{{P}_{X'Z'}} [\log q(X,Z)] \geq \EE_{{P}_{XZ}} [\log q(X,Z)] \bigg\}.
\end{flalign}
Note that, by maximizing over $P_{Z|Y}$, one can 
show that
$
R = \max_{P_{Z|Y}}R(P_{Z|Y})
$
is achievable.
\section{An Upper Bound on the Mismatch Capacity with Pre-Processing function}\label{sc: ob}
In this section an upper bound on $C_{q}^{pre}(W)$ is presented.
We use the setup described in Section \ref{sc: Problem Formulation}, and the technique for deriving an upper bound using a surely-degraded broadcast channel $W_{Y\widetilde{Y}|X}$ described in \cite{SomekhBaruch2022singleletter}.

The derivation is similar to that of \cite{SomekhBaruch2022singleletter}, in which each of the original and the auxiliary receivers will use a separate decoder, with additive metrics $q$ and $\rho$ accordingly.
 However, only the original receiver goes through a pre-processing function $f$ as can be seen in Figure \ref{fig:broadcast_up}.
\begin{figure}[htbp]
  \centering
  \includegraphics[width=0.48\textwidth]{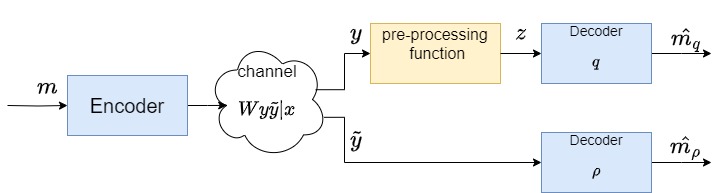}
    \caption[The channel scheme including pre-processing ]{A broadcast channel with mismatched decoders $q, \rho$ and a pre-processing function}
    \label{fig:broadcast_up}
\end{figure}
\\The input for the auxiliary decoder will be denoted as $\widetilde{Y}$ over alphabet 
$\calY$. As in \cite{SomekhBaruch2022singleletter} we define: 
 \begin{flalign}
 \Gamma& (q,\rho) = \{ P_{Z \widetilde{Y}|X}: \forall(x,z,\widetilde{y}):\spll\rho(x,\widetilde{y}) -q(x,z) < \tau_{q,\rho}(z,\widetilde{y}) =>P_{Z\widetilde{Y}|X}(z,\widetilde{y}|x) = 0 \}, \nonumber
 \end{flalign}
and
 $\tau_{q,\rho}(z,\widetilde{y}) = \max_{x' \in \cal{X} } [\rho(x',\widetilde{y}) - q(x',z)].$
 Let us consider a multicast transmission of a single message $m$ over the broadcast channel with outputs $Z,\widetilde{Y}$ and a probability function $P_{Z \widetilde{Y}|X} \in \Gamma(q,\rho)$ which satisfies 
 \begin{flalign}\label{38}
 P_{Z|X}(z|x) &= 
 \sum_{y}{P_{Z|Y}(z|y)W(y|x)} \triangleq W_{Z|X}(z|x)\end{flalign}
 It is easily verified, by definition of $\Gamma(q,\rho)$, that if the $q$ decoder at the $Z$ output successfully decodes the message, then also does the $\rho$ decoder at the $\widetilde{Y}$ output. 
\iffalse
Now, By definition of $\Gamma(q,\rho)$, if $P_{Z\widetilde{Y}|X}(\bz_i,\widetilde{\bz_i}|\bx_i) >0 $, then 
 \begin{flalign}\forall x' \in \calX : \rho(x_{i},\widetilde{y}_{i})-q(x_i,z_i) \geq \rho({x'},\widetilde{z_i})-q({x'},z_i).
 \end{flalign}
 Thus, since the channel is memoryless and the metrics are additive, if $P^n_{Z,\widetilde{Y}|X}(\bz,\widetilde{\bz}|\bx) > 0$, then $\forall \bx' \in \calX^n $ we get: \begin{flalign}\rho(\bx,\widetilde{\bz})-\rho({\bx'},\widetilde{\bz})\geq q(\bx,\bz)-q({\bx'},\bz).\end{flalign} 
Thus, letting $\calC_n = \{\bx_j\}, j \in \{1 \dots e^{nR}\}$ be a given codebook, we have that if $P^n_{Z\widetilde{Y}|X}(\bz,\widetilde{\bz}|\bx_m) > 0$ for some $\bx_m \in \calX^n$, then $\forall j \in \{1,..e^{nR}\}$: 
\begin{flalign} &\rho(\bx_m,\widetilde{\bz})-\rho({\bx_j},\widetilde{\bz})\geq q(\bx_m,\bz)-q({\bx'},\bz).
\end{flalign}
Taking the minimum over $j \neq m$ on both sides of the inequality we get that \\if $P_{Z\widetilde{Y}|X}(\bz,\widetilde{\bz}|\bx_m) > 0$, then:
\begin{flalign} 
&\rho(\bx_m,\widetilde{\bz})-\max_{j \neq m}\rho({\bx_j},\widetilde{\bz})\geq q(\bx_m,\bz)-\max_{j \neq m}q({\bx'},\bz).
\end{flalign}
This implies that given that $\bx_m$ is transmitted, if the received $\bz$ is s.t $ q(\bx_m,\bz) > \max_{j \neq m}q({\bx'},\bz)$, then necessarily $\widetilde{\bz}$ is s.t $\rho(\bx_m,\widetilde{\bz}) > \max_{j \neq m}\rho({\bx_j},\widetilde{\bz}) $ holds.
 And therefore, the error event of the $\rho$-decoder applied to channel output $\widetilde{Y}$ is contained in the error event of the  $q$-decoder applied to channel output $Z$. Thus:
 \begin{flalign}
 \forall n, \Pr (\widehat{m}_q(\bZ)=M, \widehat{m}_\rho(\widetilde{\bZ}) \neq M ) = 0
 \end{flalign}
 and consequently 
 \begin{flalign}
 C_q(W) \leq \min_{P_{Z\widetilde{Y}|X} \in \Gamma(q,\rho):P_{Z|X} = W_{Z|X}}C_\rho (P_{\widetilde{Y}|X}),
 \end{flalign}
 where  $W_{Z|X}$ is defined in (\ref{38}).
 Now, note that any additive $\rho$ gives a valid bound, so in fact, the bound can be expressed as 
 \begin{flalign}
 C_q(W) &\leq \min_{\substack{
 P_{Z\widetilde{Y}|X}:P_{Z\widetilde{Y}|X} \in \Gamma (q,\rho):  \\
 P_{Z|X} = W_{Z|X}}
 } C_{\rho} (P_{\widetilde{Y}|X})
 \\&\leq \min_{\substack{
 P_{Z\widetilde{Y}|X}:P_{Z\widetilde{Y}|X} \in \Gamma (q,\rho):  \\
 P_{Z|X} = W_{Z|X}}
 } C (P_{\widetilde{Y}|X}).
 \\& = \min_{\substack{
 P_{Z\widetilde{Y}|X}:P_{Z\widetilde{Y}|X} \in \Gamma (q,\rho):  \\
 P_{Z|X} = W_{Z|X}}
 }\max_{P_X} I(X;\widetilde{Y})
 \end{flalign}
 and as in \cite{SomekhBaruch2022singleletter}, since 
 $ \{\substack{
 P_{Z\widetilde{Y}|X}:P_{Z\widetilde{Y}|X} \in \Gamma (q,\rho):  \\
 P_{Z|X} = W_{Z|X}} \} $is a convex set, and since $ I(X;\widetilde{Y})$ is concave in $P_X$ for fixed $P_{\widetilde{Y}|X}$ and convex in  $P_{\widetilde{Y}|X}$ for a fixed $P_X$:
 \fi
Therefore (similar to \cite[Theorem 1]{SomekhBaruch2022singleletter}, when substituting the channel $W$ with $P_{Z|X}$) optimizing over the pre-processing function $f$ which may be random (and therefore denoted $P_{Z|Y}$) we obtain:
 \begin{flalign} 
C_q^{pre}(W) &\leq \max_{P_{Z|Y}} \min_{\substack{
 P_{Z\widetilde{Y}|X} \in \Gamma (q,\rho):  \\
 P_{Z|X} = W_{Z|X}}
 } C_{\rho}(P_{\widetilde{Y}|X}),
 \end{flalign}
  where $W_{Z|X}$ is defined in (\ref{38}) and depends on $P_{Z|Y}$.
 Since the bound holds for any additive $\rho(x,\widetilde{y})$, let us choose $ \rho(x,\widetilde{y}) = \log{P_{\widetilde{Y}|X} (\widetilde{y}|x)}$ in which case $C_\rho(P_{\widetilde{Y}|X}) = C(P_{\widetilde{Y}|X})$. 
 Consequently, 
 \begin{flalign}
 C_q^{pre}(W) &\leq\max_{P_{Z|Y}} \min_{\substack{
 P_{Z\widetilde{Y}|X} \in \Gamma (q,\rho):  \\
 P_{Z|X} = W_{Z|X}}
 } C (P_{\widetilde{Y}|X}).
 \\&=\max_{P_X,P_{Z|Y}} \min_{\substack{
 P_{Z\widetilde{Y}|X} \in \Gamma (q,\rho):  \\
 P_{Z|X} = W_{Z|X}}
} I(X;\widetilde{Y}).
 \end{flalign}

 Note that also sufficient conditions for the tightness of the bound, and for a metric to be capacity-achieving can be deduced from \cite[Theorem $3$, and Corollary 2]{SomekhBaruch2022singleletter}. Moreover, the upper bounds from \cite{Sphere_erexp_23} and \cite{SomekhBaruch2023UpperBoundGenie} can also be extended to this setup with the adjustment of optimizing over $P_{Z|Y}$ and substituting $W_{Y|X}$ with $W_{Z|X}$ as defined in (\ref{38}).
\section{An Optimal Pre-Processing Function which Maximizes the LM Rate}\label{ap: oppf} 
This section shows that among the symbolwise pre-processing functions there is {\it a deterministic} pre-processing function that achieves maximal rate with random coding.
This can be advantageous as there are finitely many deterministic pre-processing functions 
compared to the simplex of stochastic pre-processing functions rendering the optimization problem  significantly simpler for deterministic functions. 
For that purpose, let us define the following:
\begin{flalign}\label{def:RLM_f}
R^{pre}_{LM,q}&(W,P_{Z|Y})\triangleq \max_{\substack{P_X } } \min_{\substack{P_{X'|Z}:P_{X'}=P_X  \\ \EE q(X',Z)\geq \EE q(X,Z)}}I(X';Z).
\end{flalign} 
where, similar to (\ref{def:RLM}),  $(X,Y,Z)\sim P_X\times W\times P_{Z|Y}$. 

\begin{theorem}\label{th:deter_RLM}
There exists a deterministic pre-processing function $f: \calY \rightarrow \calY$ s.t.
\begin{flalign}
    f= \argmax_{P_{Z|Y}} R_{LM,q} ^{pre} (W,P_{Z|Y}).
\end{flalign}
 \end{theorem}
To prove the theorem we show that $R_{LM,q} ^{pre} (W,P_{Z|Y})$ is convex in $P_{Z|Y}$ and thus, achieves its maximum on the boundary of its support, therefore there exists a maximizing $P_{Z|Y}$ which is deterministic (see \ifthenelse{\boolean{bol_FULL_ARTICLE_VERSION}}{Appendix A).}{\color{blue}\color{black}\color{black}\cite[Appendix A]{ThisArticle_arxiv}),\color{black}}

\ifthenelse{\boolean{bol_FULL_ARTICLE_VERSION}}{\color{violet}\color{black}
\section{An Optimal Channel for a Given metric}\label{sc: oc}
Finding the optimal pre-processing function given a channel and a decoding metric in the symbolwise case can be viewed as a problem of finding the optimal channel under certain constraints\footnote{While we are free to choose $P_{Z|Y}$, $P_{Z|X}$ must satisfy (\ref{38}).}, for a given decoder.
This is since, as mentioned above, the concatenation of a channel and pre-processing function can be viewed as a new channel. 
In this section we analyze the unconstrained counterpart of the optimization problem for a channel with a given metric $q$ when the input alphabet $\calX$ is binary, and the output alphabet $\calY$ is a given finite set. 
	 \begin{theorem}\label{th:opt_ch2l}
The optimal channel for given alphabets $\calX,\calY$ which satisfy $|\calX|=2$ and $|\calY|<\infty$, is given by:
\begin{flalign}\label{eq:opt_ch2l}
    &W_{opt} =\spll \begin{cases} W_0 &\mbox{if } q(0,y_2)+q(1,y_1) < q(0,y_1) +q(1,y_2) \\
W_1 & \mbox{if } q(0,y_2)+q(1,y_1) > q(0,y_1) +q(1,y_2) \end{cases} 
.\end{flalign}
where:
\begin{flalign}
    W_{0}(y|x) &= \begin{cases} 1 &\mbox{if } (x,y) \in \left\{(0,y_1),(1,y_2)  \right\} \\
0 & \mbox{o.w. }  \end{cases} 
    \\W_{1}(y|x) &= \begin{cases} 1 &\mbox{if } (x,y) \in \left\{(0,y_2),(1,y_1)  \right\} \\
0 & \mbox{o.w. }  \end{cases} 
.\end{flalign}

If $q(0,y_2)+q(1,y_1) \neq q(0,y_1) +q(1,y_2)$ then
\begin{flalign}
    C_q(W_{opt}) = C(W_{opt}) = 1.
\end{flalign}
 \end{theorem}
 For the full proof see 
 \ifthenelse{\boolean{bol_FULL_ARTICLE_VERSION}}{\color{violet}\color{black}Appendix C.}{\color{blue}\color{black}\color{black}\cite[Appendix C]{ThisArticle_arxiv}\color{black}}
 \color{black}}{}
\section{A Vectorwise Pre-Processing Function}\label{sc: vpf}		
In this section we discuss the optimal pre-processing function in the case where
the output $\bz$ is constructed by the pre-processing function based on the entire channel output vector of $\by$.
This assumption in real-life applications translates to allowing a big enough latency between the input $\bx$ and output $\bz$ and a sufficiently large memory in the pre-processing function module in order to store the entire vector $\by$.\\
An extended setting compared to the one discussed here was studied in \cite{dikshtein2023mismatched} for the relay channel, the study \cite{dikshtein2023mismatched} established a lower bound on the channel capacity for a discrete or continuous Gaussian fading channel with mismatched decoder and a relay unit and provided a computationally efficient algorithm for calculating this bound.\\
We denote by $P^{vec}_e(W,\calC,q,f_n)$ the error probability achieved by using a possibly stochastic pre-processing vectorwise function $f_n$ from $\calY^n$ to $\calY^n$, 
prior to the decoder $q$, and denote the capacity for this setup as $C_q^{pre,vec}(W)$. 
\\We next define two important notions related to vectorwise pre-processing. 

\begin{definition}
\label{def: sep_ppf} 
We say that a separation principle exists for vectorwise pre-processing in $q$-mismatched channel coding, if there exists a sequence of vectorwise pre-processing functions $f_n$, $n\geq 1$,  each composed of a decoder $\varphi_n$ from $\calY^n$ to $\{1,...,M_n\}$ followed by a (possibly stochastic) mapping $g_n$ from $\{1,...,M_n\}$ to $\calY^n$ that achieves $C_q^{pre,vec}(W)$.
A pre-processing function which has this structure, will be referred to as a decode-and-process function.
\end{definition}
We next show that the separation principle holds for every channel-metric pair $(W,q)$. First we show using decode-and-process pre-processing functions one can achieve the same mismatch capacity as that of using general pre-processing functions, and subsequently we show that choosing the pre-processing function to also be deterministic is asymptotically optimal.
\begin{lemma}\label{lem:1sep}
Given codebook $\calC_n$ and a vectorwise pre-processing function $\widetilde{f}_n$ s.t $P^{vec}_e(W,\calC,q,\widetilde{f}_n) \leq \epsilon_n $, there exists a decode-and-process pre-processing function $f_n$ for which $P^{vec}_e(W,\calC_n,q,f_n) \leq 2\epsilon_n$.
 \end{lemma}
 \begin{proof}
Let $\calC_n, q$ and $\widetilde{f}_n$ from $\calY^n$ to $\calY^n$ be given, such that $P^{vec}_e(W,\calC_n,q,\widetilde{f}_n)  \leq\epsilon_n$. 
We next introduce a decode-and-process pre-processing function $f_n$ from $\calY^n$ to $ \calY^n$, consisting of two parts.
The first component is the ML decoder for channel $W$, and the second is an encoder $g$, which consists of the following parts: (A) the  encoder induced by $\calC_n$, the same codebook which is used by the transmitter, (B) the DMC $W$, and (C) the given vectorwise pre-processing function $\widetilde{f}_n$.\\
Denote by $m$,  $\widetilde{m}$, and $\hat{m}$, the transmitted message, the output of the ML decoder, and the output of the decoder $q$, respectively.
This setup is depicted in Figure \ref{fig:seperate}.

\begin{figure}[h]
		\centering
		\includegraphics[width=0.48\textwidth,height=\textheight,keepaspectratio]{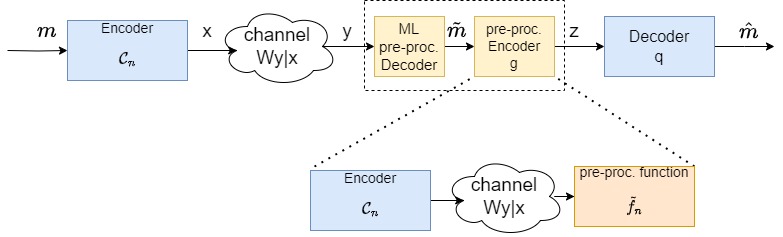}
		\caption[this figure illustrates the scheme for constructing the pre-processing function $f_n$ ]{This figure illustrates the construction for lemma 1}
		\label{fig:seperate}
	\end{figure} 

By optimality of the ML decoder, we can deduce that for the suggested scheme 
$P(\widetilde{m} \neq m) \leq P^{vec}_e(W,\calC_n,q,\widetilde{f}_n)  \leq \epsilon_n .$
  Since this choice of pre-processing function $\widetilde{f}_n$ is identical to the stochastic mapping from $\widetilde{m}$ to $\widehat{m}$, we get that 
 $P(\widehat{m} \neq \widetilde m) \leq P^{vec}_e(W,\calC_n,q,\widetilde{f}_n)  \leq \epsilon_n,$
 and consequently:
 \begin{flalign}
 P(\widehat{m} \neq m) \leq P(\widehat{m} \neq \widetilde m)+P(\widetilde{m} \neq m) \leq 2\epsilon_n.
 \end{flalign}
  Overall we get that given a vectorwise pre-processing function $\widetilde{f}_n$, we can use the suggested decode-and-process function $f_n$ described above, s.t $P^{vec}_e(W,\calC_n,q,f_n)  \leq 2\epsilon_n$.
 \end{proof}
A direct result from lemma \ref{lem:1sep} is that given a sequence of vectorwise pre-processing functions $\{\widetilde{f}_n\}$ for which $ \lim_{n\to\infty} P^{vec}_e(W,\calC_n,q,\widetilde{f}_n) = 0  $ for a given sequence of codebooks $\{\calC_n\}$, there is a sequence of decode-and-process pre-processing functions $\{f_n\}$ that, for the same codebook sequence, satisfies $ \lim_{n\to\infty}P^{vec}_e(W,\calC_n,q,f_n) = 0$.
Further, we show that deterministic decode-and-process is at least as good as decode-and-process with stochastic pre-processing.

 \begin{lemma}
For any decode-and-process pre-processing function 
$\widetilde{f}_n$ there exists a deterministic decode-and-process pre-processing function $f_n$ s.t $P^{vec}_e(W,\calC,q,f_n)  \leq P^{vec}_e (W,\calC,q,\widetilde{f}_n) $.
 \end{lemma}
 \begin{proof}
 Consider the decision region $\calD_m$ induced by the codebook $\calC_n=\{\bx_i\}$ and metric $q$ which corresponds to message index $m$; i.e., the set of all $\by$ such that  $q(\bx_m,\by)>q(\bx_j,\by)$ $ \forall j \neq m$.
 Since the decision regions $\{\calD_m\}$ are known and disjoint by definition, an optimal decode-and-process pre-processing function can be chosen s.t, each codeword $\widetilde{m}$ is mapped to a $\bz^n \in \calY^n$ which lies in $\calD_{\widetilde{m}}$, so that the decoded message $\widehat{m}$ always equals $ \widetilde{m}$ (assuming all decision regions are non-empty). 
 \end{proof}
Combining the last two lemmas we get:
\begin{theorem}\label{th: vec_sep}
For any $W$ there exists a deterministic decode-and-process pre-processing function $f_n : \calY^n \rightarrow \calY^n$ that achieves $C_q^{pre,vec}(W)$.

 \end{theorem}

\ifthenelse{\boolean{bol_FULL_ARTICLE_VERSION}}{\color{violet}\color{black}
\section{Conclusion}
\label{sec:conclusion}
The introduction of a pre-processing function makes it possible to improve on the achievable rates by adapting the channel output to better fit the mismatched decoder.
The pre-processing function could be either a vectorwise function or a symblowise function, and in this paper each of these cases is treated separately.
 For the vectorwise function, we presented a separation principle and provided a scheme for constructing a deterministic decode-and-process pre-processing function that achieves the mismatch capacity for a channel with pre-processing function.
 For the scalar channel, two single-letter lower bounds for $C_q^{pre}(W)$ were derived,
 the first bound employs a constant composition random coding scheme, and the improved bound is achieved by implementation of the superposition coding technique. A condition for the existence of a pre-processing function that achieves $C_q^{pre}(W) >0$ is derived for specific cases as well as for the case of the general DMC.
 Additionally, examples for cases where such a function exists and cases where it does not exist are provided. 
 In the case of a symbolwise pre-processing function, although the existence of a deterministic optimal function is shown for several examples, a definitive answer to whether for the general case there exists a deterministic pre-processing function that achieves the mismatch capacity, is not provided in this work.
\color{black}}{}

\ifthenelse{\boolean{bol_FULL_ARTICLE_VERSION}}{\color{violet}\color{black}
\appendix
\subsection{Proof of Theorem \ref{th:deter_RLM}}\label{appA}
For convenience denote in short:
\begin{flalign}
\psi &(P_{Z|Y}) \triangleq R^{pre}_{LM,q}(W,P_{Z|Y},P_X)= \spll \min_{\substack{P_{X'Z'}:P_{X'}=P_X,P_{Z'}=P_Z   \\ \EE \left[q(X',Z') \right] \geq \EE_{P_{XZ}} \left[q(X,Z)\right]}}I(X';Z').
 \end{flalign}
 We first show that $\psi$ is convex in $P_{Z|Y}$. Let $\alpha\in(0,1)$ be given and denote $\overline{\alpha}=1-\alpha$. Further, let 
 $P^{(\alpha)} _{Z|Y} =\alpha P^{(1)} _{Z|Y}+\bar{\alpha} P^{(2)} _{Z|Y}$, and for $\gamma \in \{1 , 2 , \alpha \}$ denote 
 \begin{flalign}
     &P^{(\gamma)} _{ZX}(z,x) = \Sigma_{y \in \calY} P^{(\gamma)} _{Z|Y}(z|y) W_{Y|X}(y|x)P_{X}(x) .
 \end{flalign} Thus,
 \begin{flalign}
\alpha \psi(P^{(1)} _{Z|Y})+&\bar{\alpha}\psi(P^{(2)} _{Z|Y}) = \spll \alpha \min_{\substack{P_{X'Z'}:P_{X'}=P_X,P_{Z'}=P^{(1)}_Z   \\ \EE \left[q(X',Z) \right] \geq \EE_{P_{ZX}^{(1)}} \left[q(X,Z)\right]}}I(X';Z') +\spll 
 \bar{\alpha} \min_{\substack{P_{X'Z'}:P_{X'}=P_X,P_{Z'}=P^{(2)}_Z   \\ \EE \left[q(X',Z) \right] \geq \EE_{P_{ZX}^{(2)}} \left[q(X,Z)\right]}}I(X';Z') .
 \end{flalign}
 We denote the corresponding minimizers by $P^{(1)} _{X'Z'} $ and $P^{(2)} _{X'Z'}$ respectively, thus:
  \begin{flalign}
\alpha \psi(P^{(1)} _{Z|Y})+\bar{\alpha}\psi(P^{(2)} _{Z|Y}) &= \alpha I(P^{(1)} _{X'Z'}) + \bar{\alpha}I(P^{(2)} _{X'Z'}) \nonumber
 \end{flalign}
 and since for fixed $P_U$, $I(P_{UV})$ is convex in $P_{V|U}$:
 \begin{flalign}
\alpha \psi(P^{(1)} _{Z|Y})+\bar{\alpha}\psi(P^{(2)} _{Z|Y}) & \geq I(\alpha P^{(1)} _{X'Z'} + \bar{\alpha}P^{(2)} _{X'Z'}).
 \end{flalign}
 Thus, to show convexity of $\psi$ we need to show that the chosen $P^{(\alpha)}_{X'Z'}$ lies in the set
 \begin{flalign}
 D^{\alpha} = \left\{ P_{Z'X'}: {\substack{P_{X'}=P_X,\\P_{Z'}=P^{(\alpha)}_Z   \\ \EE \left[q(X',Z') \right] \geq \EE_{P_{ZX}^{(\alpha)}} \left[q(X,Z)\right]}}       \right\},
 \end{flalign}
where 
 $P^{(\alpha)}_{XZ} = \alpha P^{(1)}_{XZ} +\bar{\alpha} P^{(2)}_{XZ}$.\\
 Since both $P^{(1)} _{X'Z'}, P^{(2)} _{X'Z'}$ satisfy $P^{(1)}_{X'} = P^{(2)}_{X'} = P_{X}$
 we get $P_{X'} = P^{(\alpha)}_{X'} = P_{X}$, and since  $P^{(1)} _{X'Z'}, P^{(2)} _{X'Z'}$ satisfy 
 $P_{Z'} = P^{(1)}_{Z}, P_{Z'} = P^{(2)}_{Z}$
 respectively, we get that $P_{Z'} = P^{(\alpha)} _{Z}$, thus, the first two constraints are fulfilled, and since $\EE_{P_{ZX}^{(\alpha)}}[q(X,Z)]$ is linear in $P_{Z|Y}^{(\alpha)}$, the third constraint is also fulfilled. Thus, since $\alpha P^{(1)} _{X'Z'} + \bar{\alpha}P^{(2)} _{X'Z'} \in D^\alpha ,$
 and\\ $\alpha \psi(P^{(1)} _{Z|Y})+\bar{\alpha}\psi(P^{(2)} _{Z|Y}) \geq \psi(\alpha P^{(1)} _{X'Z'} + \bar{\alpha}P^{(2)} _{X'Z'})$,
 we get that $R^{pre}_{LM,q}(W,P_{Z|Y})$ is convex in $P_{Z|Y}$ and therefore, achieves its maximum points on the boundary of its domain of definition, therefore the $P_{Z|Y}$ that yields the maximal $R^{pre}_{LM,q}(W,P_{Z|Y})$ is deterministic.
\subsection{Proof of Theorem \ref{th:useless}}\label{appB}
Let $\calX,\calY$ and $q:\calX\times\calY\rightarrow\mathbb{R}$ be given. Assume  also, that there exist $x_1,x_2 \in \calX$ and $y_1,y_2 \in \calY$ s.t 
$q(x_1,y_1)-q( x_2,y_1) \neq  q( x_1,y_2)-q( x_2,y_2).$
Denote:
\begin{flalign}
\kappa_q(& x_1,y_1,x_2,y_2) = \spll q(x_1,y_1)+q(x_2,y_2) - (q(x_1,y_2) +q(x_2,y_1))
\end{flalign}
and consider the channel:
\begin{flalign}
\overline{W}(i|j) = \begin{cases} 
1 &\mbox{if }\kappa_q( x_1,y_1,x_2,y_2) >0,\\& \ \ \  (i,j)\in \{(y_1,x_2),(y_2,x_1) \} \\
1 &\mbox{if } \kappa_q( x_1,y_1,x_2,y_2)<0,\\& \ \ \  (i,j)\in \{(y_1,x_1),(y_2,x_2) \} \\
0 &\mbox{if } j \in \{x_1, x_2\}, \  \  i\notin \{y_1,y_2\}\\
\dfrac{1}{|\calY|} &\mbox{if } j\notin \{x_1, x_2\} 
\end{cases}.
\end{flalign}
\color{blue}\color{black}
Note that using this channel exclusively for the transmission of symbols \( x \in \{x_1, x_2\} \) effectively converts it into a binary-input binary-output channel. 
In this scenario, the input and output alphabets are $\{x_1,x_2\}$ and $\{y_1,y_2\}$, respectively.
 \color{violet}\color{black}
Recall the formula given in (\ref{eq:bin_chan1})-(\ref{eq:bin_chan2}) for the mismatch capacity in the binary channel case.
\color{blue}\color{black} Thus $C_q(W) \geq 1$ for $W=\overline{W}$.\\\color{violet}\color{black}
It remains to show that if $ \forall x_1,x_2 \in \calX,\forall y_1,y_2 \in \calY$ s.t. $q( x_1,y_1)-q( x_2,y_1) = q( x_1,y_2)-q( x_2,y_2) $, 
then the metric $q$ is useless.\\ Let $\bx_0,\bx_1 \in \calX^r$
be two vectors of length $r$, and denote by $\bx_0(t), \bx_1(t)$ the entries of $\bx_0, \bx_1$ at time $t$, respectively, similarly $\by(t)$ will denote the entry of $\by$ at time $t$.
Define:
\begin{flalign}
K_{\gamma i}&=\{t:\; \bx_{\gamma}(t)=i\} :\gamma \in \{0,1\}, i \in \calX
\end{flalign}
Note that, summing over all possible $i \in \calX$, we get $ \sum_{i \in \calX}|K_{0i}|=\sum_{i \in \calX}|K_{1i}|=r$.\\For a given output vector $\by \in \calY^r$, we obtain:
\begin{flalign}
q&(\bx_0,\by)-q(\bx_1,\by) =\spll \sum_{y \in \calY , i \in \calX, t \in K_{0i}} q(i,y)\mathbbm{1}\{\by(t) = y\} \spll-  \sum_{y \in \calY , j \in \calX , t \in K_{1j}} q(j,y)\mathbbm{1}\{\by(t) = y\} =\\
& \sum_{\substack{y \in \calY , i \in \calX , j \in \calX ,\\ t \in K_{1j}\cap K_{0i}, j \neq i}} \left( q(i,y)-q(j,y) \right)\mathbbm{1}\{\by(t) = y\}
,\end{flalign}
where $\mathbbm{1}\{ \}$ is the indicator function.
From (\ref{eq:useless_cond}) we have that for any $x_1,x_2,y_1,y_2$, we get
$q(x_1,y_1)-q(x_2,y_1) = q(x_1,y_2) -q(x_2,y_2)  $,
and therefore $\forall y \in \calY $ the expression 
$q(x_1,y) -q(x_2,y)$ is independent of $y$.
Let us denote $ \delta_{q}(x_1,x_2) \triangleq q(x_1,y_1)-q(x_2,y_1) $, and this yields
\begin{flalign}
q&(\bx_0,\by)-q(\bx_1,\by) =\spll \sum_{\substack{ y \in \calY , i \in \calX , j \in \calX, \\ t \in K_{1j}\cap K_{0i}, j \neq i}} \delta_{q}(i,j)\mathbbm{1}\{\by(t) = y\}\\
&=  \sum_{i \in \calX , j \in \calX} \delta_{q}(i,j) \left|  K_{0i}\cap K_{1j}   \right|
\end{flalign}
Thus, the decoder's output does not depend on the output vector $\by^r$ and therefore does not allow passing of information, and we get $\forall W: C_q(W) = 0$ implying that it is a useless metric.

 \subsection{Proof of Theorem \ref{th:opt_ch2l}}
Assume a binary input alphabet and an $\ell$-ary output alphabet $|\calX|=2$, $|\calY|=\ell$ and a decoder with a metric $q:\calX \times \calY \rightarrow \reals$. Let us find a channel $W$ s.t $C_q(W)$ is maximal.\\
Assume $y_1, y_2 \in \calY$ s.t 
$
q(0,y_1)+q( 1,y_2) \neq  q( 1,y_1)+q( 0,y_2),
$
and denote
\begin{flalign}
  \widetilde{\kappa}_q(y_1, y_2) =  q(0,y_1)+q(1,y_2) -\left( q(0,y_2) +q(1,y_1)\right)
.\nonumber\end{flalign}
Consider the following channel:
\begin{flalign}\label{eq:chch}
W_{opt}(i|j) = \begin{cases} 
1 &\mbox{if }  \widetilde{\kappa}_q(y_1, y_2)>0,\mbox{ and }\spll \ \ \ (i,j)\in \{(y_1,1),(y_2,0) \} \\
1 &\mbox{if }  \widetilde{\kappa}_q(y_1, y_2) <0,\mbox{ and }\spll \ \ \ (i,j)\in \{(y_1,0),(y_2,1) \} \\
0 & \mbox{o.w.}  \end{cases} .
\end{flalign}
Note that the maximum capacity for a channel with binary input is $1$[bit/channel uses], and as seen in the proof of Theorem \ref{th:useless}, the mismatch capacity $C_q(W_{opt})$ equals $1$.
Thus, choosing this channel gives:
\begin{flalign}
    1= C_q(W_{opt}) \leq \max_{\widetilde{W}}C_q(\widetilde{W}) \leq \max_{\widetilde{W}}C(\widetilde{W}) = 1,
\end{flalign}
making this the optimal channel for the given decoder.
\color{black}}{}

\addcontentsline{toc}{chapter}{\sc Bibliography}
\bibliographystyle{IEEEtran}

\end{document}
